\documentclass[journal]{IEEEtran}

\ifCLASSINFOpdf

\else

\fi

\usepackage{graphicx}
\usepackage{amsthm}
\usepackage{epsfig} 
\usepackage{epstopdf}
\usepackage{amsmath}
\usepackage{amsthm}
\usepackage{comment}
\usepackage{url}
\usepackage{algorithm}
\usepackage{algorithmic}
\usepackage{subfigure}
\usepackage{booktabs}

\usepackage{amsfonts}
\usepackage{bm}
\usepackage{color}
\usepackage{cite}
\usepackage{stfloats}
\usepackage{chemarrow}
\usepackage{extarrows}
\usepackage{setspace}
\allowdisplaybreaks

\newboolean{showcomments}
\setboolean{showcomments}{true}
\newcommand{\wang}[1]{\ifthenelse{\boolean{showcomments}}
	{ \textcolor[rgb]{1,0,1}{(ZW:  #1)}}{}}
\newcommand{\fliu}[1]{\ifthenelse{\boolean{showcomments}}
	{ \textcolor{blue}{(FL:  #1)}}{}}
\newcommand{\ychen}[1]{\ifthenelse{\boolean{showcomments}}
	{ \textcolor{green}{(ZP:  #1)}}{}}
\newcommand{\slow}[1]{\ifthenelse{\boolean{showcomments}}
	{ \textcolor{blue}{(SL:  #1)}}{}}

\theoremstyle{definition}
\newtheorem{theorem}{Theorem}

\theoremstyle{definition}

\newtheorem{remark}{Remark}

\title{Optimal Emergency Frequency Control \\Based on Coordinated Droop in Multi-Infeed \\Hybrid AC-DC System}

\begin{document}
\setstretch{1}
\author{
	Ye~Liu,~\IEEEmembership{Student Member, IEEE},
	and Chen~Shen,~\IEEEmembership{Senior Member, IEEE} 
        \thanks{This work was supported  by ...  (Corresponding author: Chen Shen)   }     

		\thanks{Y. Liu and C. Shen are with China State Key Lab. of Power System, Dept. of Electrical Engineering, Tsinghua University, Beijing,
				China, 100084. e-mail: ().}
}

\markboth{Journal of \LaTeX\ Class Files,~Vol.~xx, No.~xx, xx~xxxx}%
{Shell \MakeLowercase{\textit{et al.}}: Bare Demo of IEEEtran.cls for IEEE Journals}

\maketitle

\begin{abstract}
    In multi-infeed hybrid AC-DC (MIDC) systems, the asynchronous interconnection between regional grids, the complicated system dynamics and possible cascading failures have an enormous effect on the frequency stability. In order to deal with the frequency instability problems in emergency situations, this paper proposes a decentralized emergency frequency control strategy based on coordinated droop for the MIDC system. First, a P-f droop control for LCC-HVDC systems is introduced and the coordinated droop mechanism among LCC-HVDC systems and generators is designed. Then, to reasonably allocate the power imbalance among LCC-HVDC systems and generators, an optimal emergency frequency control (OEFC) problem is formulated, and the optimal droop coefficients are selected in a decentralized approach, which can deal with various control objectives. A Lyapunov stability analysis shows that the closed-loop equilibrium is locally asymptotically stable considering the LCC-HVDC dynamics. The effectiveness of the proposed emergency control strategy is verified through simulations. 
\end{abstract}



\section{Introduction}

\subsection{Motivation and Approach}

With the continuous development of HVDC transmission technologies, conventional AC power grids have been transformed into large-scale complex hybrid AC-DC grids \cite{wang2013harmonizing}, \cite{bilodeau2016making}. In China, the implementation of enormous line-commutated-converter-based HVDC (LCC-HVDC) systems leads to asynchronous interconnection between regional power grids \cite{zhou2018principle}, and forms the multi-infeed hybrid AC-DC (MIDC) systems. In an MIDC system, multiple HVDC systems are connected with one AC system. Thus, the complicated dynamics and various faults of the MIDC system pose a serious threat to system stable operation. 

Frequency stability is of great importance for system operation. However, in MIDC systems, the conventional frequency control strategies might be difficult to ensure the frequency stability, and the reasons are: 1) the DC block faults or AC-DC cascading faults are prone to occur in MIDC systems, which could cause considerable active power imbalance. 2) Due to the feeding of multiple HVDC systems and asynchronous connections among AC systems, the system inertia and frequency regulation reserve might be not enough to meet the frequency stability requirements \cite{li2008frequency}, \cite{bevrani2009robust}. Therefore, the MIDC systems require emergency frequency control strategy in order to deal with the considerable power imbalance.

The traditional approaches of emergency frequency control are generator tripping or load shedding operations \cite{rudez2009analysis}, \cite{karady2002hybrid}, \cite{rudez2010monitoring}, but these operations will cause severe economic losses. In MIDC systems, by utilizing the fast adjustability of the HVDC system transmission power \cite{harnefors2016impact}, the more effective emergency frequency control strategy could be designed to improve the system frequency stability. In this paper, to design a decentralized approach for emergency frequency control with LCC-HVDC systems participating in, a coordinated-droop-based emergency frequency control strategy is proposed. Besides, the droop coefficients are optimized for more reasonable allocation of power imbalance.

\subsection{Literature Review}

Considering the emergency frequency control for hybrid AC-DC systems, the emergency DC power support (EDCPS) strategy is one of the effective approaches, and there are many related studies in recent years. In \cite{liu2017design}, an adaptive dynamic surface control based EDCPS strategy is proposed. In \cite{du2012integrated}, an emergency frequency control strategy considering LCC-HVDC and centre of inertia (COI) is proposed, and the adaptive backstepping sliding-mode control is utilized to guarantee its robustness. However, the above two methods are designed for synchronous AC-DC parallel interconnected systems, thus not applicable to the asynchronous MIDC systems. Towards the asynchronous MIDC system, a response-based AC-DC coordinated control strategy is proposed in \cite{shiyun2015coordinated}, which combines the EDCPS strategy and loading shedding operations. Nevertheless, the strategy in \cite{shiyun2015coordinated} is centralized with control centers, thus this strategy could not ensure a rapid response in the case of communication delay or communication failure. Obviously, the decentralized control strategies are more advantageous in emergency situations. In \cite{wenqiang2016study}, a decentralized control strategy is proposed to improve the AC frequency stability. In \cite{ma2006additional}, the frequency limit control (FLC) is proposed for hybrid AC-DC systems, which is also decentralized. However, both the strategies in \cite{wenqiang2016study} and \cite{ma2006additional} are based on the PID-type control, thus the selection and optimization for the PID parameters might be difficult in engineering practice. Moreover, all the above emergency frequency control strategies for MIDC systems do not consider the coordination with the generators' frequency regulation.

Considering the LCC-HVDC droop control proposed in this paper, droop control is a typical decentralized control strategy which is based on the relevance between two variables. To the best of our knowledge, there is no relevant research about applying droop control to LCC-HVDC systems, but droop control has been widely applied to deal with the power allocation problem in VSC-MTDC systems or microgrid systems, and there are many relevant studies on the design or optimization of droop control for this paper's reference. In \cite{abdelwahed2016power}, an optimal-power-flow based supervisory controller is designed to select the optimal droop reference voltages for VSC-MTDC systems. In \cite{li2017coordinated}, a model predictive control (MPC) based grid controller is proposed to coordinately adjust the droop gains in MTDC grids. The above two methods are centralized optimization approaches which gather the system parameters, optimize the droop coefficients and then update them to local controllers. Due to the time delay and strong communication dependence, this kind of methods cannot be applied to the control parameters optimization in emergency situations. In the field of decentralized design of droop control, the adaptive droop control (ADC) is widely utilized to design various droop characteristics under different control scenarios, e.g. , \cite{wang2017adaptive} for V-I-f droop characteristic, \cite{mortezapour2018adaptive} for interlinking converter design, and \cite{abdel2015droop} for three-wire bipolar HVDC transmission. Nevertheless, the power allocation problems of HVDC systems are barely considered in ADC strategies to optimize the droop coefficients. Besides, existing works for optimal droop control have single control objective, but relatively general optimal design methods are required due to various control objectives in various operation scenarios in engineering practice. 

\subsection{Contribution}

According to the literature review, to deal with the emergency frequency instability problems in MIDC system, the major challenge is how to design a effective decentralized emergency frequency control strategy with the LCC-HVDC participating in, and then how to reasonably allocate the power imbalance by selecting control parameters. In this paper, a simple but effective approach, i.e., the coordinated droop control is utilized to design the emergency frequency control. Besides, the power allocation problems with various allocation objectives are formulated as optimization problems, and the corresponding optimal control coefficients are selected. Benefit from decentralized design logic, the proposed control strategy can be easily applied to practical projects.

The contributions of this paper are as follows:

\begin{itemize}
	\item A coordinated-droop-based emergency frequency control strategy is proposed for the MIDC system. A P-f droop characteristic for LCC-HVDC system is introduced, and the coordinated droop mechanism among LCC-HVDC systems and generators is designed.
	\item To reasonably allocate the power imbalance among LCC-HVDC systems and generators, the optimal emergency frequency control (OEFC) problem with various control objectives is formulated. Then, the optimal droop coefficients for OEFC are given, and the optimality is proved. 
	\item Considering the LCC-HVDC dynamics, the equilibrium of the closed-loop system is analyzed and the asymptotic stability is proved through Lyapunov approach.
\end{itemize}

\subsection{Organization}
The rest of this paper is organized as follows. Section II proposes the coordinated-droop-based emergency frequency control strategy in MIDC system. Section III introduces the state model of the MIDC system. Section IV introduces the optimal droop design for power allocation. Section V discusses the stability of the closed-system equilibrium. In Section VI, a MIDC system case is tested and the effectiveness of the proposed control strategy is verified. Section VII provides the conclusion.

\section{Coordinated-Droop-Based Emergency Frequency Control Strategy}
In this section, the P-f droop control for LCC-HVDC systems is proposed. Then, the coordinated droop mechanism in MIDC system is introduced, which enables the emergency frequency control strategy.

Generally, one MIDC system can be represented as the topology shown in Fig. \ref{midc_topo}, where the AC main system contains $n_G$ synchronous generators and is connected with $n_D$ LCC-HVDC systems. There are $m$ LCC-HVDCs transmitting power from sending-end (SE) systems to the AC main system and $(n_D-m)$ LCC-HVDCs from the AC main system to receiving-end (RE) systems, which are called SE-LCC and RE-LCC systems respectively in this paper. The SE and RE systems can be collectively called the \textit{adjacent AC systems}. Based on the described MIDC system, we design the emergency frequency control strategy.

\begin{figure}[htb]
	\centering
	\includegraphics[width=0.48\textwidth]{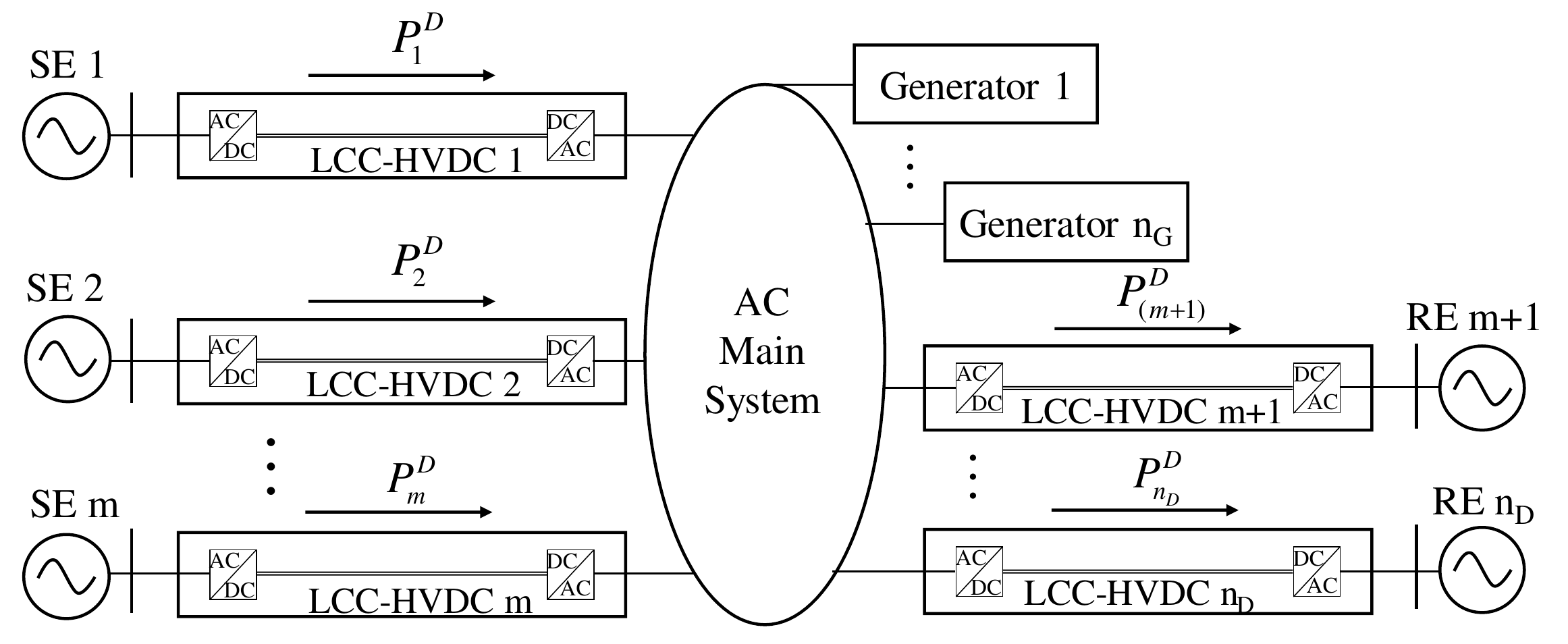}
	\caption{Topology of the MIDC system}
	\label{midc_topo}
\end{figure}

\subsection{P-f Droop Control Design for LCC-HVDC system}
The LCC-HVDC system is widely applied in the power grids due to its considerable transmission capacity and long transmission distance. The operation state of LCC-HVDC system can be determined by the static operation characteristic \cite{arrillaga1998high}, \cite{kundur1994power}, as shown in Fig. \ref{lcc_static_fig}.

\begin{figure}[htb]
	\centering
	\includegraphics[width=0.28\textwidth]{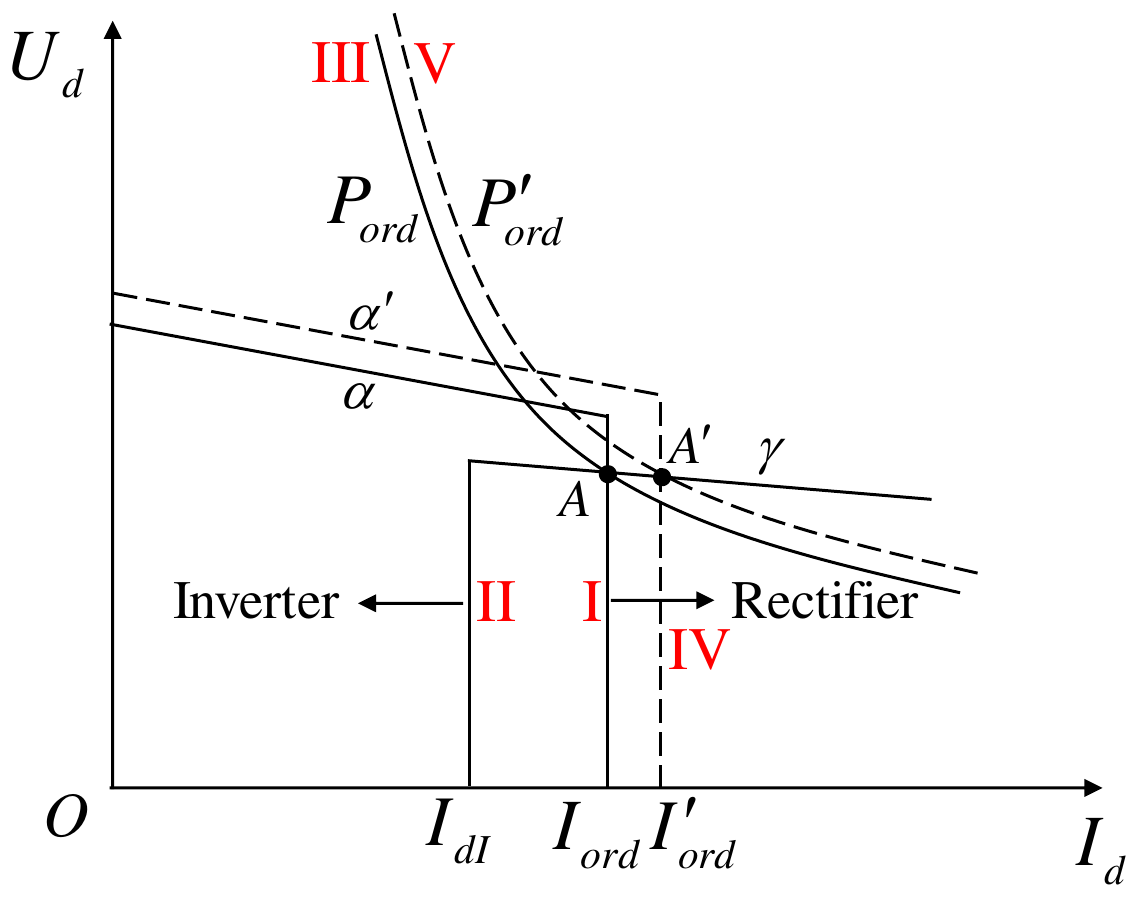}
	\caption{The static operation characteristic of LCC-HVDC}
	\label{lcc_static_fig}
\end{figure}

The normal operation curves of the LCC-HVDC are shown as I and II in Fig. \ref{lcc_static_fig}, where the vertical parts of the curves represent constant current control and the upper parts are constant $\alpha$ control and constant $\gamma$ control respectively. The intersection point $A$ indicates that the rectifier of the LCC-HVDC adopts constant current control while the inverter adopts constant $\gamma$ control during the normal operation. Besides, the constant power control can be added to the rectifier to regulate the transmission power directly, which is shown as curve III in Fig. \ref{lcc_static_fig}. Curve III also passes the point $A$. We have:
\begin{align}
I_{ord}=\frac{P_{ord}}{U_{d}}
\end{align}
where $I_{ord}$ and $P_{ord}$ are the current order and power order for the rectifier control, $U_d$ is the DC voltage.

In this section, we design the droop characteristic between the frequency of the AC system and the active power of the LCC-HVDC, i.e., the P-f droop. Suppose that the $P_{ord}$ increases to $P_{ord}^{'}$ due to frequency fluctuation, the $I_{ord}$ will also increase to $I_{ord}^{'}$ and the firing angle $\alpha$ will decrease to $\alpha^{'}$ as shown in Fig. \ref{lcc_static_fig}. The above changes make the operation point $A$ change to $A^{'}$, curve I change to IV, and curve III change to V. In addition, the DC voltage will drop slightly due to the constant $\gamma$ control of the inverter. Note that the firing angle $\alpha$ has limits to ensure the normal operation, thus the $P_{ord}$ also has limits, i.e., $P_{ord} \in \Omega_P$, which should be considered in the design of LCC-HVDC droop.

Generally, in the MIDC system as shown in Fig. \ref{midc_topo}, one LCC-HVDC system can participate in the frequency control in both the SE and RE system. Therefore, the P-f droop control equations are shown as follows:
\begin{subequations}
	\label{lcc_droop_equ}
	\begin{align}
	&P_{ord}=P_{dN} - k_{droop}^{re}(\omega_{re}-\omega_{reN}),\ \text{for RE}
	\label{lcc_droop_equ_1}\\
	&P_{ord}=P_{dN} + k_{droop}^{se}(\omega_{se}-\omega_{seN}),\ \text{for SE}\label{lcc_droop_equ_2}\\
	&(P_{ord} \in \Omega_P\nonumber)
	\end{align}
\end{subequations}
where $k_{droop}^{re}$ and $k_{droop}^{se}$ are the selected droop coefficients for RE and SE respectively, $\omega_{re}$ and $\omega_{se}$ are the AC frequencies, and the subscript $N$ represents the nominal values. Note that in the presence of the private communication network specifically built for control and protection in hybrid AC-DC power grid, the frequency signal of the RE system can be rapidly transmitted to the bipolar power control station located in the SE side, which enables the P-f droop control for the RE system. Considering the droop control for both SE and RE systems, the control framework for LCC-HVDC system is shown in Fig. \ref{lcc_droop_frame}.

\begin{figure}[htb]
	\centering
	\includegraphics[width=0.48\textwidth]{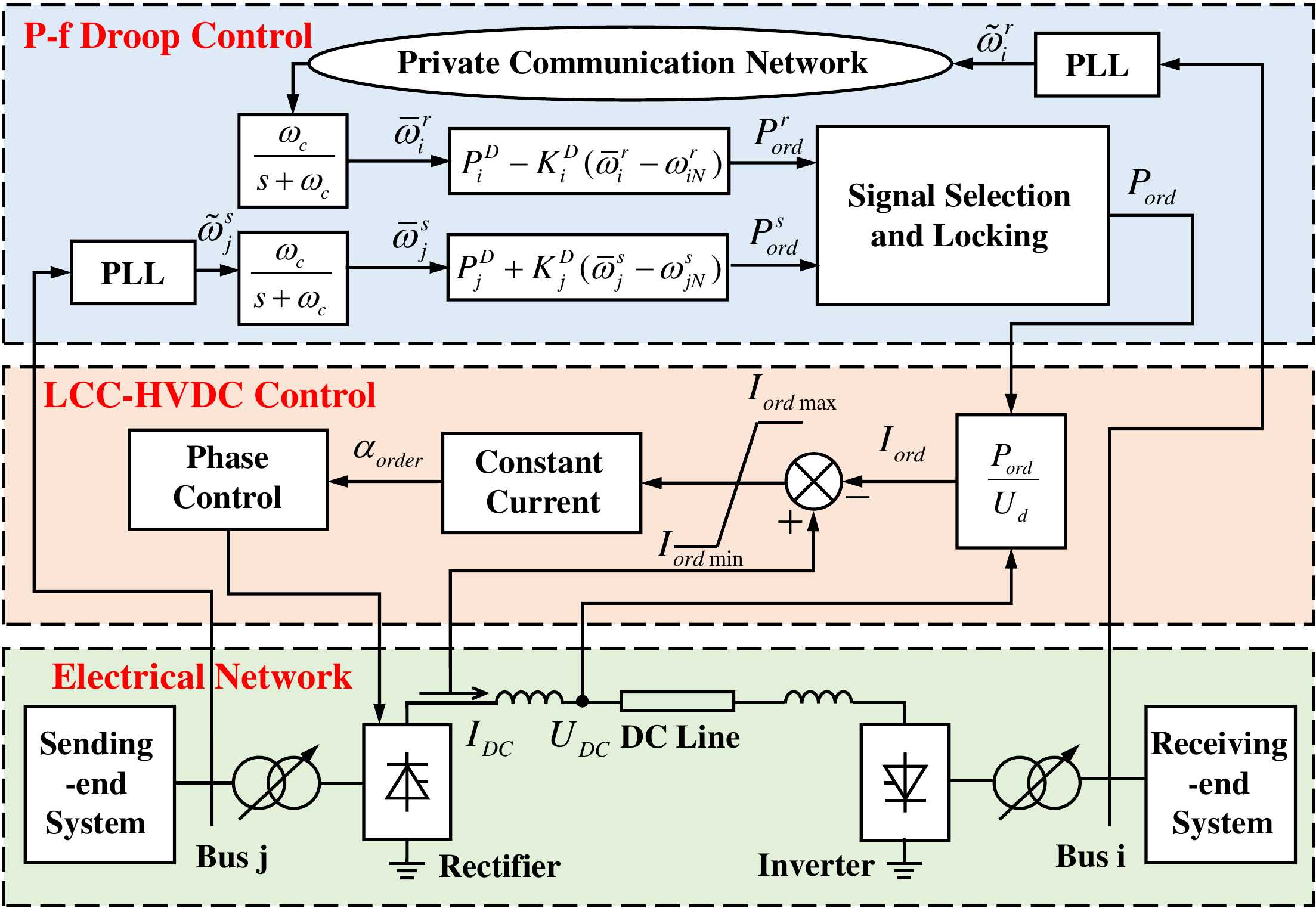}
	\caption{The control framework for LCC-HVDC}
	\label{lcc_droop_frame}
\end{figure}

In Fig. \ref{lcc_droop_frame}, if there is frequency drop in the RE system due to some faults, the DC power will increase according to the RE droop. Then, this DC power regulation will also result in the frequency loss at SE system, which will enable the SE droop and affect the frequency control in the RE system. Thus, the SE droop and RE droop cannot work at the same time. To solve the aforementioned problem, a signal selection and locking module is adopted as shown in Fig. \ref{lcc_droop_frame}, which will output the earlier-responding power regulation signal ($p_{ord}^r$ or $p_{ord}^s$) and lock the other signal. Then, the selected $P_{ord}$ signal is transferred to the HVDC control, and the phase control signal will be output to the electrical system.

\begin{remark}
	During normal operation, the LCC-HVDC systems are supposed to transmit specific powers as planned. Thus, the LCC-HVDC droop control should not work when the system operates at normal state. Besides, the coordinated droop mechanism in Section II.B avoids the mal-operating of LCC-HVDC systems. In this control framework, the private communication network plays an important role which passes the frequency signal of the RE system to the SE side.
\end{remark}

\subsection{Coordinated Droop Mechanism for Emergency Frequency Control}

In power grids, the synchronous generators are usually implemented with droop control, which is realized mainly by the governors and participates in the primary frequency control. However, the designed LCC-HVDC droop control do not participate in the conventional primary frequency control since the DC power should be kept constant without emergencies. Thus, the LCC-HVDC droop control and the generators' primary droop are relatively independent of each other, and a coordinated droop mechanism is proposed to make the LCC-HVDC droop as support for primary frequency modulation in emergency situations, which further formulates the emergency frequency control strategy. 

The proposed coordinated droop mechanism and the emergency frequency control strategy for the MIDC system are shown in Fig. \ref{coor_efc}, and this mechanism mainly contains the following two points: 

\begin{figure}[htb]
	\centering
	\includegraphics[width=0.46\textwidth]{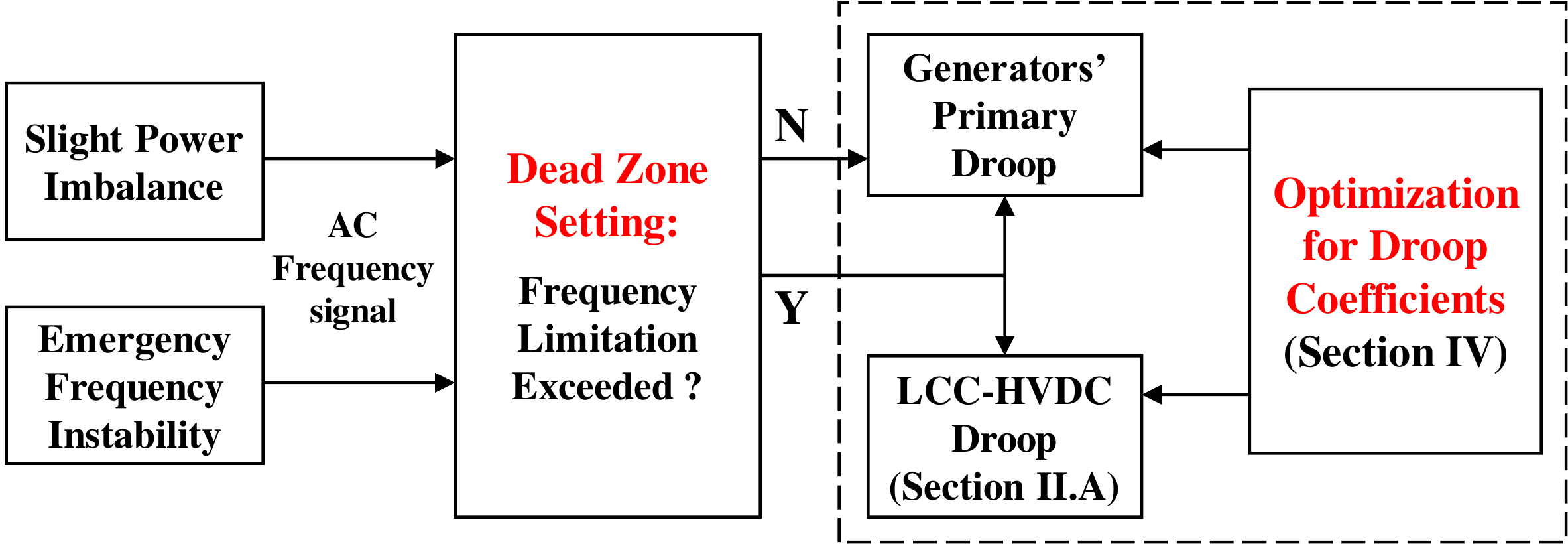}
	\caption{Coordinated-droop-based emergency frequency control strategy}
	\label{coor_efc}
\end{figure}

\textit{1) Dead Zone Setting}. Since the generators' primary droop always works even with tiny frequency fluctuations while the LCC-HVDC droop works only in case of emergency frequency problems, a dead zone setting for LCC-HVDC droop control is necessary. There are two common methods to set the dead zone, i.e., the frequency deviation limitation and the frequency change rate limitation, and in this paper we adopt the former. When the system frequency changes due to some faults, the frequency limitation of dead zone is utilized to determine whether there is an emergency and whether to enable the LCC-HVDC droop.

\textit{2) Optimization for Droop Coefficients}. The coordinated optimization for the droop coefficients will be introduced in detail in Section IV, where the optimal emergency frequency control (OEFC) problem is formulated to reasonably allocate power imbalance among LCC-HVDC systems and generators during the control process. The optimal droop coefficients stay constant during the operation and are updated only if the control objective of the OEFC problem changes.

\begin{remark}
	The designed LCC-HVDC droop control is decentralized because the local frequency signal is fed back to the controller, and the optimization of droop coefficients is also carried out in a decentralized approach. Thus, the coordinated-droop-based emergency frequency control strategy is decentralized.  
\end{remark}

\section{State Model of MIDC system}
In this section, the state model of the MIDC system shown in Fig. \ref{midc_topo} is introduced for optimal control design and stability analysis. We ignore the dynamics of the adjacent AC systems (defined in  Section II) in order to focus on the AC main system and its connected LCC-HVDC systems. We also ignore the dead zone setting of LCC-HVDC droop in state modeling due to its nonlinearity, and the reasonability will be verified in the case study. 

Generally, a power system can be described by a graph $\mathcal{G}=(\mathcal{N},\mathcal{E})$, where the nodes represent the buses denoted by $\mathcal{N}$ and the edges represent the transmission lines denoted by $\mathcal{E} \subseteq \mathcal{N} \times \mathcal{N}$. In this paper, the system contains three types of buses, i.e., the generator buses, the LCC-HVDC connected buses and the passive load buses, which are denoted by $\mathcal{N}_G$, $\mathcal{N}_D$ and $\mathcal{N}_P$ respectively. Thus we have $\mathcal{N}=\mathcal{N}_G \cup \mathcal{N}_D \cup \mathcal{N}_P$. Denote the numbers of buses in above sets by $n$, $n_G$, $n_D$ and $n_P$, then we have $n= n_G + n_D + n_P$. We ignore the load shedding operations since the LCC-HVDCs and generators can provide enough power support. We have the following assumptions:
\begin{itemize}
	\item One LCC-HVDC connected bus cannot be a generator bus, i.e., $\mathcal{N}_G \cap \mathcal{N}_D = \emptyset$.
	\item One bus in $\mathcal{N}_D$ is only connected with one LCC-HVDC line.
\end{itemize}

With the above assumptions, considering the coordinated-droop-based emergency frequency control strategy among LCC-HVDCs and generators, the MIDC system in Fig. \ref{midc_topo} can be described by the differential algebraic equations (DAEs) shown in \eqref{system_model_daes}, where we consider the second-order dynamic models of generators and the first-order inertia models of LCC-HVDC systems.
\begin{subequations}
	\label{system_model_daes}
	\begin{align}
	&\dot{\theta}_i=\omega_i,\ i \in \mathcal{N}_G \cup \mathcal{N}_D
	\label{system_model_daes_2}\\
	&M_i \dot{\omega}_i +D_i \omega_i = P_i - \sum_{j\in \mathcal{N}} B_{ij} \sin(\theta_i - \theta_j) - \overline{k}_i^G \omega_i,\ i \in \mathcal{N}_G
	\label{system_model_daes_1}\\
	&0=P_i + p_i^{dc} - \sum_{j\in \mathcal{N}} B_{ij} \sin(\theta_i - \theta_j),\ i \in \mathcal{N}_D
	\label{system_model_daes_3}\\
	&0=P_i - \sum_{j\in \mathcal{N}} B_{ij} \sin(\theta_i - \theta_j),\ i \in \mathcal{N}_P\label{system_model_daes_4}\\
	&T_i^D \dot{p}_i^{dc} = -p_i^{dc}+P_i^D-k_i^D \omega_i,\ i \in \mathcal{N}_D \label{system_model_daes_5}
	\end{align}
\end{subequations}
where $\theta_i$ is the phase angle at bus $i$ with reference to synchronous rotation coordinate, $\omega_i$ is the frequency deviation from the nominal frequency, $M_i>0$ is the inertia constant of the generator $i$, $D_i>0$ is the damping coefficient, $P_i$ is the power injection ($>0$) or demand ($<0$), $p_i^{dc}$ is the transmission power of LCC-HVDC $i$ which $>0$ when the AC main system is RE or $<0$ when this system is SE, $P_i^D$ is the nominal value of $p_i^{dc}$, $T_i^D$ is the inertia time constant of LCC-HVDC $i$, $B_{ij}=\overline{B}_{ij} V_i V_j$ is the effective susceptance of line $(i,j)$, $V_i$ is the voltage amplitude at bus $i$ which is assumed to be constant due to its irrelevance with the frequency control, $\overline{k}_i^G > 0$ is the droop coefficient of the generator $i$, and $k_i^D > 0$ is the droop coefficient of the LCC-HVDC $i$. Note that \eqref{system_model_daes_5} can represent the droop control equations both for RE and SE since we have defined the sign of $p_i^{dc}$, thus $k_i^D$ is the combination of $k_{droop}^{re}$ and $k_{droop}^{se}$ in \eqref{lcc_droop_equ}.

Define the effective droop coefficient $k_i^G = \overline{k}_i^G + D_i$ for generator $i$, then \eqref{system_model_daes_1} can be represented as:
\begin{align}
\label{system_model_daes_1_new}
M_i \dot{\omega}_i = P_i - \sum_{j\in \mathcal{N}} B_{ij} \sin(\theta_i - \theta_j) - k_i^G \omega_i,\ i \in \mathcal{N}_G
\end{align}

When designing the optimal emergency frequency control in Section IV, we ignore the dynamics of LCC-HVDC system because this dynamics have no effect on the optimal steady state. In this situation, let $\dot{p}_i^{dc}=0$ and combining \eqref{system_model_daes_3} and \eqref{system_model_daes_5}, we have:
\begin{align}
\label{system_model_eq_1}
0=P_i + P_i^{D} - \sum_{j\in \mathcal{N}} B_{ij} \sin(\theta_i - \theta_j)-k_i^D \omega_i,\ i \in \mathcal{N}_D
\end{align}

\section{Optimal Droop Design for Power Imbalance Allocation}
In this section, the power imbalance allocation problem is formulated as the optimal emergency frequency control problem with various control objectives. Then, the optimal droop coefficients for LCC-HVDCs and generators are selected and the optimality is proved.

\subsection{Optimal Emergency Frequency Control Problem}

In this MIDC system, there are various AC-DC faults which could cause the power imbalance. The power imbalance can be represented as $\sum_{i\in \mathcal{N}} P_i +\sum_{i\in \mathcal{N}_D} P_i^D \not=0$. The optimal emergency frequency control (OEFC) problem is formulated to reasonably allocate the power imbalance during the control process. In this section, the various control cost functions are defined according to various control objectives in engineering practice. Then, the reasonable power imbalance allocation can be achieved by minimizing the total control cost. Define the power regulations by the droop control strategy as $u_i^G = -k_i^G \omega_i$ for generator $i$ and $u_i^D = -k_i^D \omega_i$ for LCC-HVDC $i$. Then, we can derive the general method to formulate the OEFC problem by selecting two different control objectives. The dead zone has no effect on the system's steady-state solutions and thus we ignore it during the optimal design. 

\textit{1) Control objective I}: The LCC-HVDC with larger regulation margin provides more power support.
For generator $i$, we define the cost function $C_i^G (u_i^G)$ in a classic form \cite{wood2013power}:
\begin{align}
	\label{gene_cost}
	C_i^G (u_i^G) = \frac{1}{2} \beta_i (u_i^G)^2
\end{align}
where $\beta_i$ is the cost coefficient for generator $i$. For LCC-HVDC systems, due to the differences among the upper bounds, lower bounds and nominal values of the transmission power of LCC-HVDC systems, we define the power regulation margin of LCC-HVDC $i$ as:
\begin{align}
Z_{i}^{D}=\left\{\begin{array}{l}{\overline{P}_{i}^{D}-P_{i}^{D}},\ \text{when power increases} \\ {P_{i}^{D}-\underline{P}_{i}^{D}},\ \text{when power decreases}
\end{array}\right.
\end{align}
Then, to let the LCC-HVDC system with larger power regulation margin provide more power support for the power imbalance of the whole system, we define the cost function of LCC-HVDC $i$ as:
\begin{align}
C_i^D (u_i^D)=\alpha_i (\frac{u_{i}^{D}}{Z_{i}^{D}})^{2}=\frac{\alpha_i}{(Z_{i}^{D})^{2}}(u_{i}^{D})^{2}
\end{align}
where $\alpha_i$ is the cost coefficient for LCC-HVDC $i$. Under control objective I, the total control cost to minimize is:
\begin{align}
\label{total_cost_1}
\sum_{i\in \mathcal{N}_G} \frac{1}{2} \beta_i (u_i^G)^2 + \sum_{i\in \mathcal{N}_D} \frac{\alpha_i}{(Z_{i}^{D})^{2}}(u_{i}^{D})^{2}
\end{align}

\textit{2) Control objective II}: The adjacent AC systems have equal frequency deviations during the emergency frequency control. The generators also adopt the cost function as \eqref{gene_cost}. To have equal frequency deviations, the adjacent AC system with larger primary frequency modulation coefficient should provide more power support, i.e., its connected HVDC system should provide more power support. Thus, we define the cost function of LCC-HVDC $i$ as:
\begin{align}
C_i^D(u_i^D)=e_i (\Delta \omega_i^{'})^2=e_{i}(\frac{u_{i}^{D}}{K_{i}^{f}})^{2}=\frac{e_{i}}{(K_{i}^{f})^{2}}(u_{i}^{D})^{2}
\end{align}
where $\Delta \omega_i^{'}=\omega_i^{'}-\omega_N^{'}$ is the frequency deviation of adjacent AC system $i$, $e_i$ is the cost coefficient and $K_i^f$ is the primary frequency modulation coefficient of adjacent AC system $i$. In engineering practice, there exist multiple HVDC systems connected to the same adjacent AC system, and we can make them equivalent to one HVDC system and then optimize the control. The total control cost under objective II is:
\begin{align}
\label{total_cost_2}
\sum_{i \in N_{G}} \frac{1}{2} \beta_{i}(u_{i}^{G})^{2}+\sum_{i \in \mathcal{N}_{D}} \frac{e_{i}}{(K_{i}^{f})^{2}}(u_{i}^{D})^{2}
\end{align}

Note that \eqref{total_cost_1} and \eqref{total_cost_2} have the same form, thus the later theoretical analysis takes control objective I as an example and these two objectives will be discussed in the case study in Section VI. Generally, the optimal design method in this paper is applicable so long as the cost function of the HVDC system can be described as the quadratic form of DC power regulation in physical sense.

In summary, the general OEFC problem is as follows:
\begin{subequations}
	\label{oefc_pro}
	\begin{align}
	\min_{u_i^D \in \Omega_i^D,u_i^G}\ &\sum_{i\in \mathcal{N}_G} C_i^G (u_i^G)+\sum_{i\in \mathcal{N}_D} C_i^D (u_i^D) \nonumber\\
	=&\sum_{i\in \mathcal{N}_G} \frac{1}{2} \beta_i (u_i^G)^2 + \sum_{i\in \mathcal{N}_D} \frac{\alpha_i}{\left(Z_{i}^{D}\right)^{2}}\left(u_{i}^{D}\right)^{2}
	\\
	\text{s.t.}\ \sum_{i \in \mathcal{N}} P_{i}&+\sum_{i \in \mathcal{N}_{D}} P_{i}^{D}+\sum_{i \in \mathcal{N}_{G}} u_{i}^{G}+\sum_{i \in \mathcal{N}_{D}} u_{i}^{D}=0
	\end{align}
\end{subequations}
where the constraints is the DC power limitations and the power balance of the whole system. We have the following assumption.

\textbf{Assumption 1.} The OEFC problem \eqref{oefc_pro} is feasible, and the global optimal solution of $u_i^D$ is located in its feasible region $\Omega_i^D=\left\{u_i^D \big|\underline{u}_i^D \le u_i^D \le \overline{u}_i^D \right\}$, where $\underline{u}_i^D$ and $\overline{u}_i^D$ are the lower and upper bounds.

\subsection{Optimal Droop Coefficients and Optimality Analysis}

To select the droop coefficients, we have the following theorem.

\begin{theorem}
	If Assumption 1 holds, we have:
	\begin{enumerate}
		\item the optimal droop coefficients are:
		\begin{subequations}
			\begin{align}
			&k_i^G = \frac{1}{\beta_i},\ i \in \mathcal{N}_G
			\label{opt_droop_1}\\
			&k_i^D = \frac{(Z_i^D)^2}{2 \alpha_i},\ i \in \mathcal{N}_D \label{opt_droop_2}
			\end{align}
		\end{subequations} 
	which are the solutions of the OEFC problem \eqref{oefc_pro}.
		\item with the optimal droop coefficients setting, the dynamic of the whole system is equivalent to a partial primal-dual distributed algorithm, which guarantees the optimality.
	\end{enumerate}
\end{theorem}

\begin{proof}
	Firstly, from \eqref{oefc_pro}, we can derive the objective function of the dual problem of OEFC problem, i.e., the Lagrangian dual function \cite{boyd2004convex}:
	\begin{align}
	\Phi(\lambda)=\inf_{u_i^D \in \Omega_i^D,u_i^G}\left(\sum_{i \in \mathcal{N}_{G}}\left(\frac{1}{2} \beta_i \left(u_i^G\right)^2+\lambda u_{i}^{G}\right)\nonumber\right.\\
	+\left.\sum_{i \in \mathcal{N}_{D}}\left(\frac{\alpha_i}{\left(Z_{i}^{D}\right)^{2}}\left(u_{i}^{D}\right)^{2}+\lambda P_{i}^{D}+\lambda u_{i}^{D}\right)+ \sum_{i \in \mathcal{N}} \lambda P_{i}\right)
	\end{align}
	where $\lambda$ is the dual variable. Then, we can solve the infimum problem explicitly, the results are as follows.
	\begin{subequations}
		\label{results_1}
		\begin{align}
		\Phi(\lambda)=& \sum_{i \in \mathcal{N}_{D}}\left(\frac{\alpha_i}{\left(Z_{i}^{D}\right)^{2}}\left(u_{i}^{D}(\lambda)\right)^{2}+\lambda P_{i}^{D}+\lambda u_{i}^{D}(\lambda)\right) \nonumber\\
		+&\sum_{i \in \mathcal{N}_{G}} \left(-\frac{1}{2\beta_i}\lambda^2\right)+ \sum_{i \in \mathcal{N}} \lambda P_{i}\\
		u_i^G =& -\frac{1}{\beta_i}\lambda\\
		u_i^D(\lambda) =& \left[-\frac{(Z_i^D)^2}{2 \alpha_i} \lambda \right]_{\Omega_i^D}	
	\end{align}
	\end{subequations}
	
	From \eqref{results_1}, the solution of the OEFC problem requires the communication among the buses due to the common variable $\lambda$. In this paper, to solve it in a distributed or decentralized approach, we define the vector $\bm{\lambda}=\{\lambda_i, i\in \mathcal{N}\}$, where each $\lambda_i$ corresponds to the bus $i$. At the optimal solution, there is $\lambda_i = \lambda_j, (i,j)\in \mathcal{E}$. The dual problem of OEFC (DOEFC) is:
	\begin{subequations}
		\label{doefc_pro}
		\begin{align}
		\max_{\bm{\lambda}}\Phi(\bm{\lambda})&=\sum_{i \in \mathcal{N}_{D}}\left(\frac{\alpha_i}{\left(Z_{i}^{D}\right)^{2}}\left(u_{i}^{D}(\lambda_i)\right)^{2}+\lambda_i P_{i}^{D}+\lambda_i u_{i}^{D}(\lambda_i)\right) \nonumber\\
		&+\sum_{i \in \mathcal{N}_{G}} \left(-\frac{1}{2\beta_i}\lambda_i^2\right)+ \sum_{i \in \mathcal{N}} \lambda_i P_{i}\\
		\text{s.t.}\ \lambda_i =& \lambda_j,\ (i,j)\in \mathcal{E}
		\end{align}
	\end{subequations}

The Lagrangian function of the DOEFC problem is:
\begin{align}
L(\bm{\lambda},\bm{\nu}):=\Phi(\bm{\lambda})+\sum_{(i,j) \in \mathcal{E}} \nu_{i j}\left(\lambda_{i}-\lambda_{j}\right)
\end{align}
where the vector $\bm{\nu}=\{\nu_{ij},(i,j)\in \mathcal{E}\}$ is the Lagrangian multiplier.	Then, under Assumption 1, which fields $u_i^D(\lambda)= -\frac{(Z_i^D)^2}{2 \alpha_i} \lambda$, we apply the partial primal-dual distributed algorithm to solve the DOEFC problem, which take the form:
\begin{subequations}
	\begin{align}
	\dot{\lambda}_{i}&=\tau_{i} \frac{\partial L(\bm{\lambda}, \bm{\nu})}{\partial \lambda_{i}}=\tau_{i}\left(-\frac{1}{\beta_i} \lambda_{i}+P_{i}+\sum_{j \in \mathcal{N}} \nu_{i j}\right), i \in \mathcal{N}_{G}
	\label{distri_algo_1}\\
	0&=\frac{\partial L(\bm{\lambda}, \bm{\nu})}{\partial \lambda_{i}}=P_{i}^{D}+u_{i}^{D}\left(\lambda_{i}\right)+P_{i}+\sum_{j \in \mathcal{N}} \nu_{i j},\ i \in \mathcal{N}_{D}
	\label{distri_algo_2}\\
	0&=\frac{\partial L(\bm{\lambda}, \bm{\nu})}{\partial \lambda_{i}}=P_{i}+\sum_{j \in \mathcal{N}} \nu_{i j},\ i \in \mathcal{N}_{P}
	\label{distri_algo_3}\\
	\dot{v}_{i j}&=-\gamma_{i j} \frac{\partial L(\bm{\lambda}, \bm{\nu})}{\partial \nu_{ij}}=-\gamma_{i j}\left(\lambda_{i}-\lambda_{j}\right),\ (i, j) \in \mathcal{E} \label{distri_algo_4}
	\end{align}
\end{subequations}
where $\tau_{i}>0$, $\gamma_{ij}>0$ are the stepsizes. If we identify $\lambda_i$ with $\omega_i$ and $\nu_{ij}$ with $P_{ij}$, set the stepsizes $\tau_i = \frac{1}{M_i}$, $\gamma_{ij} = B_{ij} \cos (\theta_i - \theta_j) $, integrate \eqref{distri_algo_4} and obtain:
\begin{align}
P_{ij} = -B_{ij} \sin(\theta_i - \theta_j)
\end{align}
Then, \eqref{distri_algo_1}-\eqref{distri_algo_4} are identical to the closed-loop system dynamic \eqref{system_model_daes_2}, \eqref{system_model_daes_1}, \eqref{system_model_daes_4} and \eqref{system_model_eq_1}, which means that the system will reach the optimization in a distributed algorithm approach and ensures the optimality.
\end{proof}

\section{Stability Analysis}

In this section, we analyze the stability of the closed-loop system with the designed coordinated-droop-based emergency frequency control strategy by the Lyapunov approach. We rewrite the closed-loop system state equations as \eqref{system_model_re}.
\begin{subequations}
	\label{system_model_re}
	\begin{align}
	&\dot{\theta}_i=\omega_i,\ i \in \mathcal{N}_G \cup \mathcal{N}_D
	\label{system_model_re_1}\\
	&M_i \dot{\omega}_i = P_i - \sum_{j\in \mathcal{N}} B_{ij} \sin(\theta_{ij}) - k_i^G \omega_i,\ i \in \mathcal{N}_G
	\label{system_model_re_2}\\
	&0=P_i + p_i^{dc} - \sum_{j\in \mathcal{N}} B_{ij} \sin(\theta_{ij}),\ i \in \mathcal{N}_D
	\label{system_model_re_3}\\
	&0=P_i - \sum_{j\in \mathcal{N}} B_{ij} \sin(\theta_{ij}),\ i \in \mathcal{N}_P\label{system_model_re_4}\\
	&T_i^D \dot{p}_i^{dc} = -p_i^{dc}+P_i^D-k_i^D \omega_i,\ i \in \mathcal{N}_D \label{system_model_re_5}\\
	&k_i^G = \frac{1}{\beta_i},\ i \in \mathcal{N}_G\\
	&k_i^D = \frac{(Z_i^D)^2}{2 \alpha_i},\ i \in \mathcal{N}_D
	\end{align}
\end{subequations}
where $\theta_{ij}=\theta_i - \theta_j$ is the phase angle difference between the connected nodes $i$ and $j$. 

Then, we denote the column vectors $\theta_G = \{ \theta_i,i \in \mathcal{N}_G\}$, $\theta_D = \{ \theta_i,i \in \mathcal{N}_D\}$, $\theta_P = \{ \theta_i,i \in \mathcal{N}_P\}$, $\omega_G = \{ \omega_i,i \in \mathcal{N}_G\}$, $\omega_D = \{ \omega_i,i \in \mathcal{N}_D\}$, $p^{dc}=\{p^{dc}_i,i \in \mathcal{N}_D\}$ and $\theta=(\theta_G^T,\theta_D^T,\theta_P^T)^T$, $\omega=(\omega_G^T,\omega_D^T)^T$. 

We have the following assumption.

\textbf{Assumption 2.} Due to the dead zone setting, the initial post-fault state is different from the state that the LCC-HVDC droop control starts working. We ignore the effect of dead zone and discuss it in case study. We suppose that the adjacent AC systems are stable and their stabilities are not discussed.
 
\textbf{Assumption 3.} The equilibrium point $(\theta^*,\omega_G^*,p^{dc*})$ of the closed-loop system \eqref{system_model_re} satisfies:
\begin{align}
\theta^* \in \Theta=\left\{\theta\in \mathbb{R}^n \Big||\theta_i - \theta_j|<\frac{\pi}{2},\forall (i,j) \in \mathcal{E}\right\}
\end{align}

In practice, Assumption 3 is usually considered as a security constraint for power flow solutions. Then, we have the following theorem for the stability of the system \eqref{system_model_re}.


\begin{theorem}
	If Assumption 2-3 hold, for the closed-loop system \eqref{system_model_re}, the following statements hold:
	\begin{enumerate}
		\item there exists an equilibrium point $(\theta^*,\omega_G^*,p^{dc*}) \in \Psi=\mathbb{R}^n \times \mathbb{R}^{n_G}\times \mathbb{R}^{n_D}$.
		\item there exist a domain $\Psi^s \subset \Psi$ such that for any initial state $(\theta^0,\omega_G^0,p^{dc0}) \in \Psi^s$ that satisfies the DAEs \eqref{system_model_re}, the state trajectory converges to the unique equilibrium state $(\theta^*,\omega_G^*,p^{dc*})\in \Psi$.
	\end{enumerate}
\end{theorem}

\begin{proof}
	1) When the system is at the equilibrium state, we have $\dot{\omega}_i=0,i \in \mathcal{N}_G$, $\dot{p}_i^{dc}=0,i \in \mathcal{N}_D$ and $\omega_i^*=\omega_j^*=\omega_{syn}, \forall i, j \in \mathcal{N}$. Then, sum \eqref{system_model_re_2}-\eqref{system_model_re_4} and we have the following equations at the equilibrium point:
	\begin{subequations}
		\begin{align}
	\sum_{i\in\mathcal{N}} P_i &+ \sum_{i\in\mathcal{N}_D} p_i^{dc*}-\sum_{i\in\mathcal{N}_G} k_i^G \omega_{i}^*=0\\
	&p_{i}^{dc*}=P_{i}^{D}-k_{i}^{D} \omega_{i}^{*},\ i\in \mathcal{N}_D
	\end{align}
	\end{subequations}
	which yields:
	\begin{align}
	\omega_{i}^*=\omega_{syn}=\frac{\sum\limits_{i\in\mathcal{N}} P_i + \sum\limits_{i\in\mathcal{N}_D} P_i^D}{\sum\limits_{i\in\mathcal{N}_G} k_i^G + \sum\limits_{i\in\mathcal{N}_D} k_i^D}
	\end{align}
	
	Note that when there exists power imbalance in the hybrid AC-DC system, $\sum_{i\in\mathcal{N}} P_i + \sum_{i\in\mathcal{N}_D} P_i^D \not=0$, which yields  $\omega_{i}^*\not=0,i \in \mathcal{N}_G \cup \mathcal{N}_D$, thus $\theta_i^*$ is time variant and the equilibrium $(\theta^*,\omega^*,p^{dc*})$ may not be unique. Nevertheless, the differences $\theta_i^*-\theta_j^*, (i,j)\in\mathcal{E}$, i.e., the relative phase angles are constant. It follows from \cite{skar1980stability} and \cite{araposthatis1981analysis} that there exists at most one power flow solution which satisfies $\theta\in\Theta$ in the lossless system \eqref{system_model_re}, thus the equilibrium is unique under the definition of relative phase angle.
	
	2) Before the stability analysis of the equilibrium, we first prove the regularity of algebraic equations \eqref{system_model_re_4}, which also ensures the existence and of the solution of \eqref{system_model_re}. We define the following function with respect to $\theta$:
	\begin{align}
	F_{c}(\theta)=\sum_{(i, j) \in \mathcal{E}} B_{i j}\left(1-\cos \left(\theta_{i}-\theta_{j}\right)\right)
	\end{align}
	The Hessian matrix $H_c(\theta)$ of $F_{c}(\theta)$ is:
	\begin{align}
	\label{h_c}
	\small
	\left[\begin{array}{cccc}
	\bar{B}_{11}(\theta) & -B_{12} \cos \left(\theta_{12}\right) & \cdots & -B_{1 n} \cos \left(\theta_{1 n}\right) \\
	-B_{21} \cos \left(\theta_{21}\right) & \bar{B}_{22}(\theta) & \cdots & -B_{2 n} \cos \left(\theta_{2 n}\right) \\
	\vdots & \vdots & \ddots & \vdots \\
	-B_{n 1} \cos \left(\theta_{n 1}\right) & -B_{n 2} \cos \left(\theta_{n 2}\right) & \cdots & \bar{B}_{n n}(\theta)
	\end{array}\right]
	\end{align}
	where $\bar{B}_{i i}(\theta)=\sum_{j \in \mathcal{N}} B_{i j} \cos \left(\theta_{i j}\right)$ and $\theta_{ij}=\theta_i - \theta_j$. The $H_c(\theta)$ is the Laplacian matrix of the graph $\mathcal{G}$ with line weights $B_{ij}\cos(\theta_i - \theta_j)$, and the line weights are positive Under the Assumption 3. Thus, $H_c(\theta)$ is positive definite and its principle minors are all nonsingular \cite{brualdi1991combinatorial}. Notice that the Jacobian matrix of \eqref{system_model_re_4} is the principle minor of $H_c(\theta)$ and thus nonsingular. Therefore, the algebraic equations \eqref{system_model_re_4} are regular.
	
	Then, we prove the asymptotic stability of the equilibrium $(\theta^*,\omega^*,p^{dc*})$ by the Lyapunov approach. 
	
	Consider the following Lyapunov function candidate:
	\begin{align}
	&V(\theta, \omega_{G}, p^{dc})=V_{1}+V_{2}\\
	&V_{1}\left(\theta, \omega_{G}\right)=F_{c}(\theta)-F_{c}\left(\theta^{*}\right)-\nabla_{\theta}^{T} F_{c}\left(\theta^{*}\right)\left(\theta-\theta^{*}\right)\nonumber\\&\qquad\qquad\ \ \ +\frac{1}{2}\left(\omega_{G}-\omega_{G}^{*}\right)^{T} M_{G}\left(\omega_{G}-\omega_{G}^{*}\right) \nonumber\\
	&V_{2}\left(p^{dc}\right)=\sum_{i \in \mathcal{N}_{D}}\left(\frac{1}{2} d_{i} T_{i}^{D}\left(p_{i}^{dc}-p_{i}^{dc*}\right)^{2}\right)\nonumber
	\end{align}
	where $M_G=diag(M_i),i\in \mathcal{N}_G$, and $d_i>0$ are constant.
	
	Taking time derivative of $V_1$, we have:
	\begin{align}
	\label{lyapunuv_1_1}
	\dot{V}_1=\left(\nabla_{\theta}^{T} F_{c}(\theta)-\nabla_{\theta}^{T} F_{c}\left(\theta^{*}\right)\right) \dot{\theta}+\left(\omega_{G}-\omega_{G}^{*}\right)^{T} M_{G} \dot{\omega}_{G}
	\end{align}
	where
	\begin{align}
	&\nabla_{\theta} F_{c}(\theta)-\nabla_{\theta} F_{c}\left(\theta^{*}\right)\nonumber\\
	&=\left[\begin{array}{l}
	\mathop{col}\limits_{i \in \mathcal{N}_{G}}\left(\sum\limits_{j \in N} B_{i j} \sin \left(\theta_{ij}\right)-\sum\limits_{j \in N} B_{i j} \sin \left(\theta_{ij}^{*}\right)\right) \\
	\mathop{col}\limits_{i \in \mathcal{N}_{D}}\left(\sum\limits_{j \in N} B_{i j} \sin \left(\theta_{ij}\right)-\sum\limits_{j \in N} B_{i j} \sin \left(\theta_{ij}^{*}\right)\right) \\
	\mathop{col}\limits_{i \in \mathcal{N}_{P}}\left(\sum\limits_{j \in \mathcal{N}} B_{i j} \sin \left(\theta_{ij}\right)-\sum\limits_{j \in \mathcal{N}} B_{i j} \sin \left(\theta_{ij}^{*}\right)\right)
	\end{array}\right]\nonumber\\
	&=\left[\begin{array}{l}
	\mathop{col}\limits_{i \in \mathcal{N}_{G}}\left(-k_{i}^{G} (\omega_{i}-\omega_{i}^{*})-M_{i} \dot{\omega}_{i}\right) \\
	\mathop{col}\limits_{i \in \mathcal{N}_{D}}\left(p_{i}^{dc}-p_{i}^{dc*}\right) \\
	\mathop{col}\limits_{i \in \mathcal{N}_{P}}\left(0\right)
	\end{array}\right]
	\label{lyapunuv_1_2}
	\end{align}
	Substitute \eqref{lyapunuv_1_2} into \eqref{lyapunuv_1_1} and we get:
	\begin{align}
	\label{v_1_1}
	\dot{V}_1 =& \sum_{i \in \mathcal{N}_{G}}\left(-k_{i}^{G} (\omega_{i}^{2}-\omega_{i}^{*} \omega_{i})\right)+\sum_{i \in \mathcal{N}_{D}}\left(\left(p_{i}^{dc}-p_{i}^{dc*}\right) \omega_{i}\right)\nonumber\\
	&-\sum_{i \in \mathcal{N}_{G}}\left(M_{i} \omega_{i}^{*} \dot{\omega}_{i}\right)
	\end{align}
	Due to $\omega_i^*=\omega_j^*=\omega_{syn}$, the third term of \eqref{v_1_1} can be represented as:
	\begin{align}
	&\sum_{i \in \mathcal{N}_{G}}\left(M_{i} \omega_{i}^{*} \dot{\omega}_{i}\right)=\omega_{i}^{*}\sum_{i \in \mathcal{N}_{G}}\left(M_{i}  \dot{\omega}_{i}\right)\nonumber\\
	&=\omega_{i}^{*}\left(\sum_{i \in \mathcal{N}} P_{i}+\sum_{i \in \mathcal{N}_{D}} p_{i}^{dc}-\sum_{i \in \mathcal{N}_{G}} k_{i}^{G} \omega_{i}\right)\nonumber\\
	&=\sum_{i \in \mathcal{N}_{G}}\left(k_{i}^{G}\left(\left(\omega_{i}^{*}\right)^{2}-\omega_{i} \omega_{i}^{*}\right)\right)+\sum_{i \in \mathcal{N}_{D}}\left(\omega_{i}^{*}( p_{i}^{dc}- p_{i}^{dc*})\right)
	\label{v_1_2}
	\end{align}
	Substitute \eqref{v_1_2} into \eqref{v_1_1}, we have:
	\begin{align}
	\dot{V}_1 =&-\sum_{i \in \mathcal{N}_{G}}\left(k_{i}^{G}\left(\omega_{i}-\omega_{i}^{*}\right)^{2}\right)\nonumber\\
	&+\sum_{i \in \mathcal{N}_{D}}\left(\left(p_{i}^{dc}-p_{i}^{dc*}\right)\left(\omega_{i}-\omega_{i}^{*}\right)\right)
	\end{align}
	
	Consider $V_{2}\left(p^{dc}\right)$, we have:
	\begin{align}
	&\dot{V}_2 = \sum_{i \in \mathcal{N}_{D}}\left(d_{i} T_{i}^{D}\left(p_{i}^{dc}-p_{i}^{dc*}\right) \dot{p}_{i}^{dc}\right)\nonumber\\
	&=\sum_{i \in \mathcal{N}_{D}}\left(d_{i}\left(p_{i}^{dc}-p_{i}^{dc*}\right)\left(-p_{i}^{dc}+p_{i}^{dc*}+k_{i}^{D} \omega_{i}^{*}-k_{i}^{D} \omega_{i}\right)\right)\nonumber\\
	&=\sum_{i \in \mathcal{N}_{D}}\left(-d_{i}\left(p_{i}^{dc}-p_{i}^{dc*}\right)^{2}-d_{i} k_{i}^{D}\left(p_{i}^{dc}-p_{i}^{dc*}\right)\left(\omega_{i}-\omega_{i}^{*}\right)\right)\nonumber
	\end{align}
	
	Hence, we can derive:
	\begin{align}
	\dot{V}=&-\sum_{i \in \mathcal{N}_{G}}\left(k_{i}^{G}\left(\omega_{i}-\omega_{i}^{*}\right)^{2}\right)-\sum_{i \in \mathcal{N}_{D}}\left(d_{i}\left(p_{i}^{dc}-p_{i}^{dc*}\right)^{2}\right)\nonumber\\
	&+(1-d_i k_i^D)\sum_{i \in \mathcal{N}_{D}}\left(\left(p_{i}^{dc}-p_{i}^{dc*}\right)\left(\omega_{i}-\omega_{i}^{*}\right)\right)
	\end{align}
	
	Therefore, by selecting $d_i={1}/{k_i^D}>0,i\in \mathcal{N}_D$, we obtain $\dot{V}\le 0$, and $\dot{V}=0$ only at the equilibrium.
	
	Then, we prove that the Lyapunov function candidate $V(\theta, \omega_{G}, p^{dc})\ge 0$ and the equilibrium $z^*=(\theta^*,\omega^*,p^{dc*})$ is a strict minimum point of $V$. It can be easily verified that $V|_{z^*}=0$. For $\nabla V|_{z^*}$, we have:
	\begin{align}
	&\nabla V|_{z^*}=col(\nabla_\theta V,\nabla_{\omega_G} V,\nabla_{p^{dc}} V)|_{z^*}\nonumber\\
	&\nabla_\theta V = \nabla_{\theta} F_{c}(\theta)-\nabla_{\theta} F_{c}\left(\theta^{*}\right)\nonumber\\
	&\nabla_{\omega_G} V =  M_G \left(\omega_G - \omega_G^* \right)\nonumber\\
	&\nabla_{p^{dc}} V = d^D T^D (p^{dc}-p^{dc*})
	\end{align}
	We can obtain that $\nabla V|_{z^*}=\bm{0}\in\mathbb{R}^{n+n_G+n_D}$. Then, for $\nabla^2 V$, we have:
	\begin{align}
	\nabla^2 V = diag(H_c,H_G,H_D)
	\end{align}
	where $\nabla^2 V$ is a block diagonal matrix, $H_c$ is positive definite according to \eqref{h_c}. $H_G$ and $H_D$ are respectively the Hessian matrices of $\frac{1}{2}\left(\omega_{G}-\omega_{G}^{*}\right)^{T} M_{G}\left(\omega_{G}-\omega_{G}^{*}\right)$ and $V_2$, which are both positive functions, thus $H_G$ and $H_D$ are also positive definite. Therefore, we have $\nabla^2 V\succ 0$ and equilibrium $z^*$ is a strict minimum point of $V$.
	
	Moreover, it's obvious to see that the invariant set $\left\{(\theta,\omega_G,p^{dc})\big|\dot{V}(\theta,\omega_G,p^{dc})=0\right\}$ contains only the equilibrium point. In summary, the Lyapunov function candidate $V(\theta, \omega_{G}, p^{dc})$ satisfies the stability criterion \cite{hill1990stability}, which completes the proof.
	
\end{proof}

\begin{remark}
	According to the stability analysis, the equilibrium point is asymptotically stable so long as the droop coefficients $k_i^D>0, i\in \mathcal{N}_D$, which is a sufficient condition. The size of the attraction domain is not discussed in this paper, which we will focus on in the future work.
\end{remark}

\section{Case Study}

In this section, the effectiveness of the proposed coordinated-droop-based emergency frequency control, the optimality of the selected droop coefficients and the system stability are illustrated by a case study on the CloudPSS platform \cite{song2019cloudpss}, \cite{cloudpss}.

\subsection{Test System Description}

The topology of the MIDC test system is shown in Fig. \ref{case_topo}, which is a modified IEEE New England system combining the CIGRE HVDC benchmark systems \cite{faruque2005detailed}. The full electromagnetic transient (EMT) model of the test system is built on the CloudPSS platform \cite{liu2018modeling}. The main AC system is connected with four $\pm$660kV monopolar 12-pulse CIGRE LCC-HVDC systems, in which the rectifier adopts constant power control while the inverter adopts constant $\gamma$ control. The power transmission directions are shown in Fig. \ref{case_topo}, i.e., the LCC1, LCC2 and LCC3 are SE-LCC systems while LCC4 is RE-LCC system for the AC main system. The adjacent AC systems adopt the equivalent centre of inertia (COI) model. The LCC-HVDC systems are implemented with the designed P-f droop control, while the seven generators (equivalent from generator units) are implemented with primary droop control.

\begin{figure}[htb]
	\centering
	\includegraphics[width=0.48\textwidth]{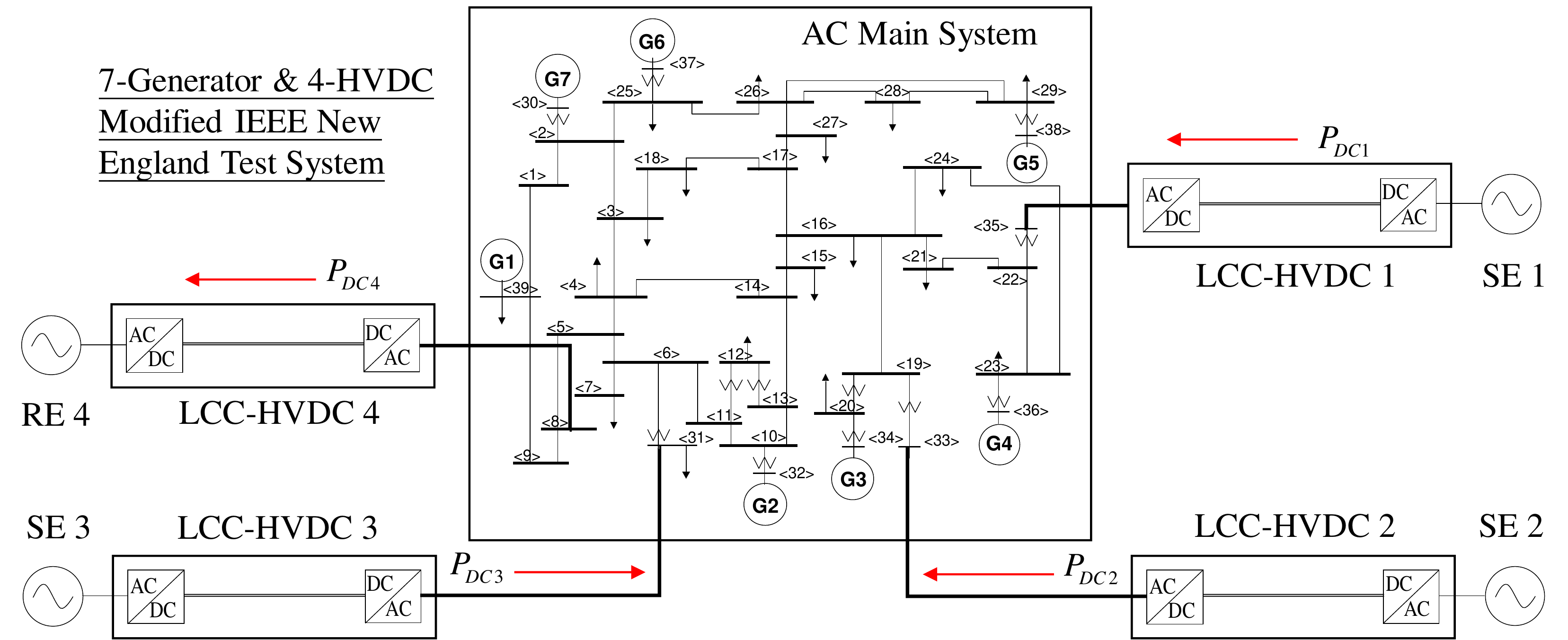}
	\caption{Topology of the MIDC test system}
	\label{case_topo}
\end{figure}

Set the active power base-value as $P_B=100MW$, and set the cost coefficients of generators $\beta_1=\beta_5=\beta_6=0.1p.u.$ and $\beta_2=\beta_3=\beta_4=\beta_7=0.2p.u.$. For LCC-HVDC systems, the related parameters under control objective I and II are shown in Table I.


\begin{table}[htbp] 
	\footnotesize
	\centering
	\vspace{-0.2cm}  %
	\setlength{\abovecaptionskip}{0.cm}
	\setlength{\belowcaptionskip}{-0.cm}
	\caption{Related Parameters of LCC-HVDC Systems}
	\label{tab1} 
	\begin{tabular}{cccccc} 
		\toprule 
		No. & $P_i^D$ & $\overline{P}_i^D$, $\underline{P}_i^D$ & $\alpha_i$ & $K_i^f$ & $e_i$\\ 
		\midrule 
		LCC1 & 645 MW & 750,550 MW & 0.05 p.u. & 25 p.u. & 30 p.u.\\
		LCC2 & 630 MW & 750,550 MW & 0.05 p.u. & 30 p.u. & 30 p.u.\\
		LCC3 & 660 MW & 750,550 MW & 0.05 p.u. & 20 p.u. & 30 p.u.\\
		LCC4 & 500 MW & 600,400 MW & 0.05 p.u. & 25 p.u. & 30 p.u.\\  
		\bottomrule 
	\end{tabular} 
\end{table}

The optimal droop coefficients are obtained by \eqref{opt_droop_1} and \eqref{opt_droop_2}, and under Control Objective II, the optimal droop coefficients for LCC-HVDCs are:
\begin{align}
k_i^D=\frac{(K_i^f)^2}{2 e_i},i \in \mathcal{N}_D
\end{align} 
Then we define the \textit{average droop coefficients}, i.e., the average value of optimal droop coefficients, and respectively set them for LCC-HVDC systems and generators as the contrast group. For the generators, the optimal droop coefficients are: $k_1^G=k_5^G=k_6^G=10p.u.$, $k_2^G=k_3^G=k_4^G=k_7^G=5p.u.$, and the average ones are all $7.14p.u.$. For the LCC-HVDC systems, the optimal and average droop coefficients under Control Objective I and II are shown in Table II.

\begin{table}[htbp] 
	\footnotesize
	\centering
	\vspace{-0.2cm}  %
	\setlength{\abovecaptionskip}{0.cm}
	\setlength{\belowcaptionskip}{-0.cm}
	\caption{Optimal Droop coefficients of LCC-HVDC Systems}
	\label{tab2} 
	\begin{tabular}{ccccc} 
		\toprule 
		 & \multicolumn{2}{c}{Objective I} & \multicolumn{2}{c}{Objective II}\\
		\midrule
		No. & Opt. droop & Ave. droop & Opt. droop & Ave. droop \\ 
		\midrule 
		LCC1 & 11.03 p.u. & 10.88 p.u.& 10.42 p.u. & 10.63 p.u. \\
		LCC2 & 14.40 p.u. & 10.88 p.u. & 15.00 p.u. & 10.63 p.u. \\
		LCC3 & 8.10 p.u. & 10.88 p.u. & 6.67 p.u. & 10.63 p.u. \\
		LCC4 & 10.00 p.u. & 10.88 p.u. & 10.42 p.u. & 10.63 p.u. \\  
		\bottomrule 
	\end{tabular} 
\end{table}

We set generator G6 trip at the time of 8s, which causes 530MW power imbalance (approximate 10\% of the system capacity) and can be considered as emergency situation. Based on the above settings, we carry on the following simulations and analyses 

\subsection{Effectiveness of Coordinated-Droop-Based Emergency frequency Control}

To verify the effectiveness of the proposed control strategy, we set the following three subcases: (1) generators have droop control with optimal coefficients while LCC-HVDCs have no droop control. (2) all generators and LCC-HVDCs have droop control with optimal coefficients under Control Objective I. (3) on the basis of (2), the LCC-HVDC droop have dead zone with 49.8Hz limit. Then, the frequencies of AC main system and the active powers of LCC-HVDCs and generators are shown in Fig. \ref{case_1} and Fig. \ref{case_2} respectively.

\begin{figure}[htb]
	\centering
	\includegraphics[width=0.48\textwidth]{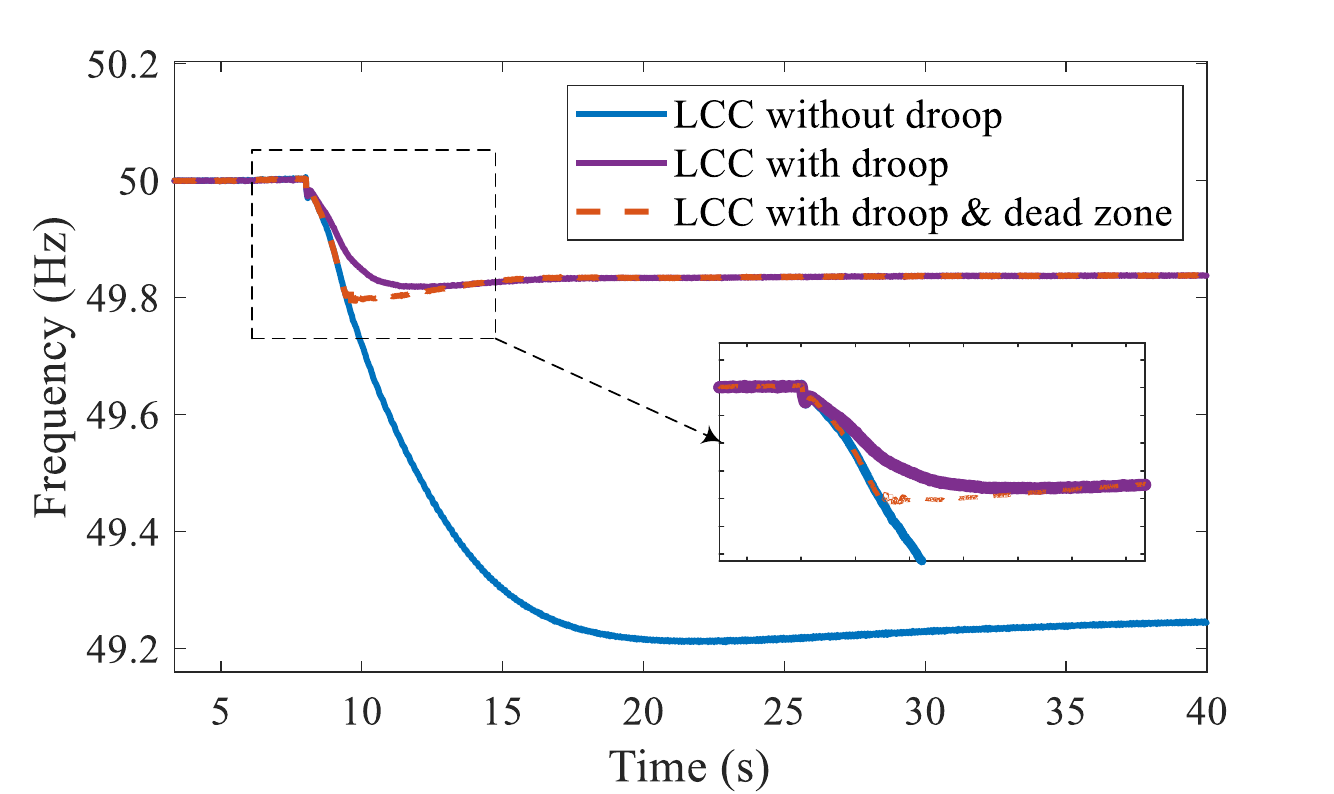}
	\caption{Frequencies of AC main system}
	\label{case_1}
\end{figure}

\begin{figure}[htb]
	\centering
	\includegraphics[width=0.48\textwidth]{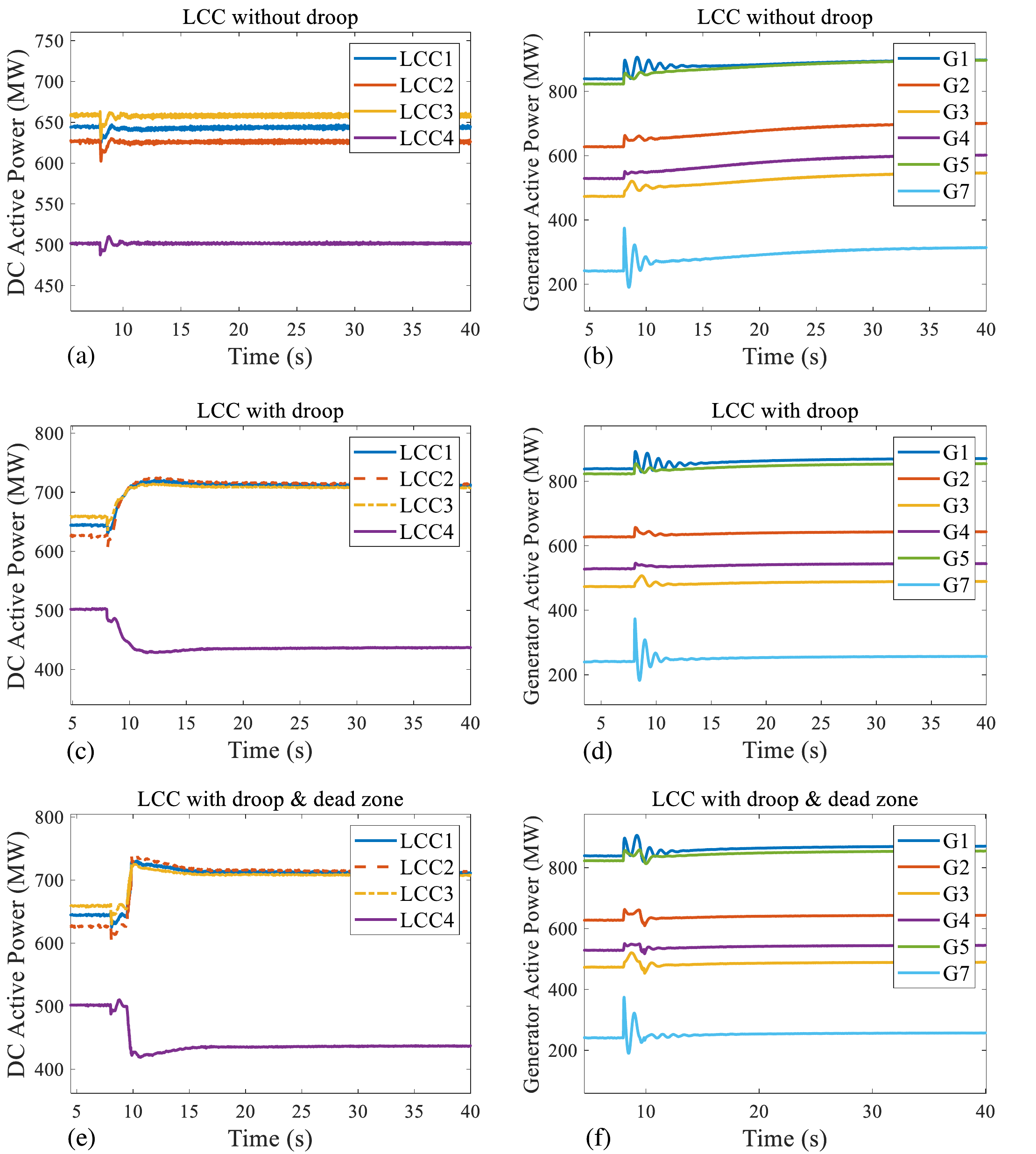}
	\caption{Active powers of LCC-HVDCs and generators. (a)(b) LCC-HVDCs have no droop control. (c)(d) LCC-HVDCs have droop control. (e)(f) LCC-HVDCs have droop control and dead zone setting.}
	\label{case_2}
\end{figure}

In Fig. \ref{case_1}, when LCC-HVDC systems have no droop control, the AC main system frequency reduces to approximate 49.25Hz at about 30s. In engineering practice, this low frequency has already caused severe instability. When the LCC-HVDCs and generators are all equipped with droop control, the system frequency stabilizes at approximate 49.85Hz at about 15s. Compared with subcase (1), subcase (2) has shorter transient time and the steady-state frequency of subcase (2) is closer to nominal frequency. Thus, the proposed coordinated-droop-based emergency frequency control strategy is effective. Further, we add the dead zone to LCC-HVDC droop control in subcase (3), and we find that the dead zone setting has no effect on the steady-state frequency, and has tiny influence on the transient frequency process which can be ignored. Therefore, it is reasonable to ignore the dead zone setting in the optimal control design and stability analysis process. As shown in Fig. \ref{case_2}, in subcase (1), the emergency frequency regulation can only rely on the generators' primary droop, but the power adjustment speed of generators is relatively slow. In subcase (2) and (3), the fast transmission power adjustability of LCC-HVDC systems is utilized to provide considerable power support to the AC main system and relieve the frequency modulation pressure of the generators, which can also verifies the effectiveness of proposed coordinated control strategy. By comparing Fig. \ref{case_2}(c)(d) and Fig. \ref{case_2}(e)(f), we can also show that the dead zone has little effect on the control performance.

Moreover, we output the DC variables curves of LCC-HVDC 3 in subcase (2) as shown in Fig. \ref{case_3}. With the increase of the transmission active power, the DC current also increases while the DC voltage decreases slightly, and the $\alpha$ angle decreases during the transient process. These change trends are consistent with the theoretical analysis in section II.A and verify the correctness of designed P-f droop characteristic.

\begin{figure}[htb]
	\centering
	\includegraphics[width=0.48\textwidth]{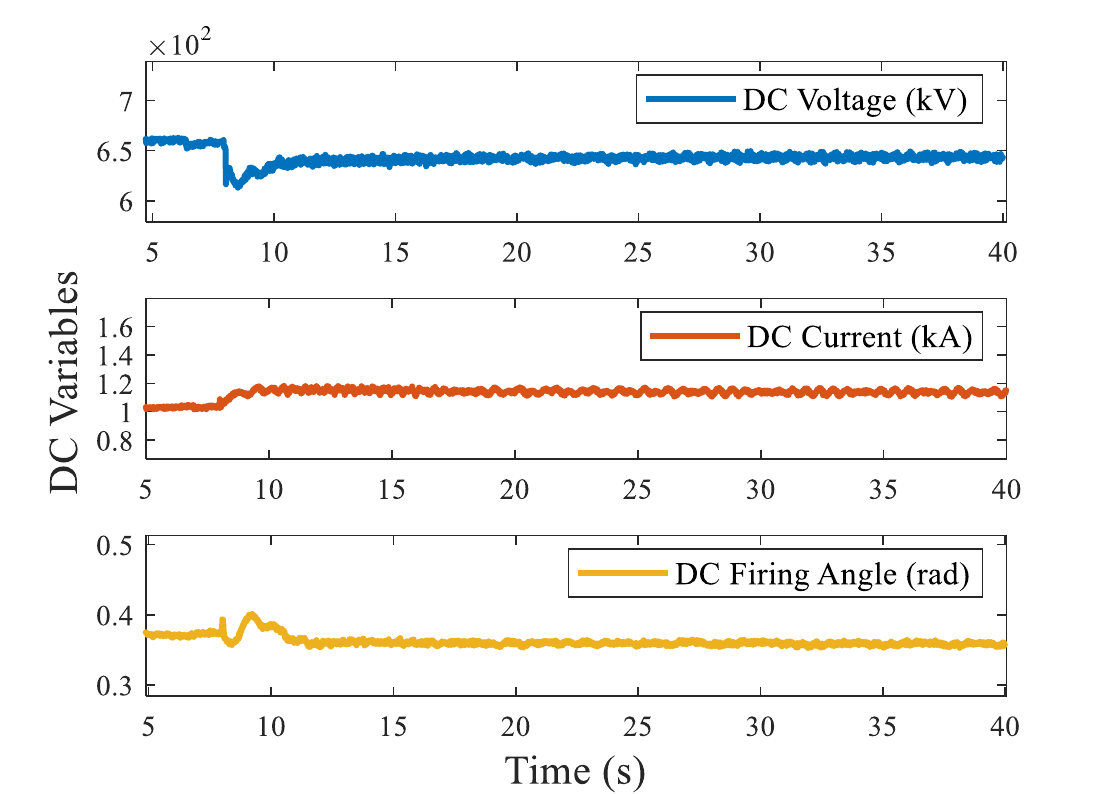}
	\caption{DC variables of LCC-HVDC3}
	\label{case_3}
\end{figure}

\subsection{Optimality of Droop Coefficients}

To reasonably allocate the power imbalance among LCC-HVDC systems and generators during control process, the optimality of the selected droop coefficients is verified. Four groups of simulations are carried out, which adopt the optimal droop coefficients (Opt. droop) and average droop coefficients (Ave. droop) under Control objective I and II respectively. The system frequencies all stabilize at above 49.8Hz with these four groups of coefficients (not shown here), which guarantees the frequency stability of the MIDC system. Then, we output the active powers of LCC-HVDCs and calculate the total control costs, as shown in Fig. \ref{case_4}.

\begin{figure}[htb]
	\centering
	\includegraphics[width=0.48\textwidth]{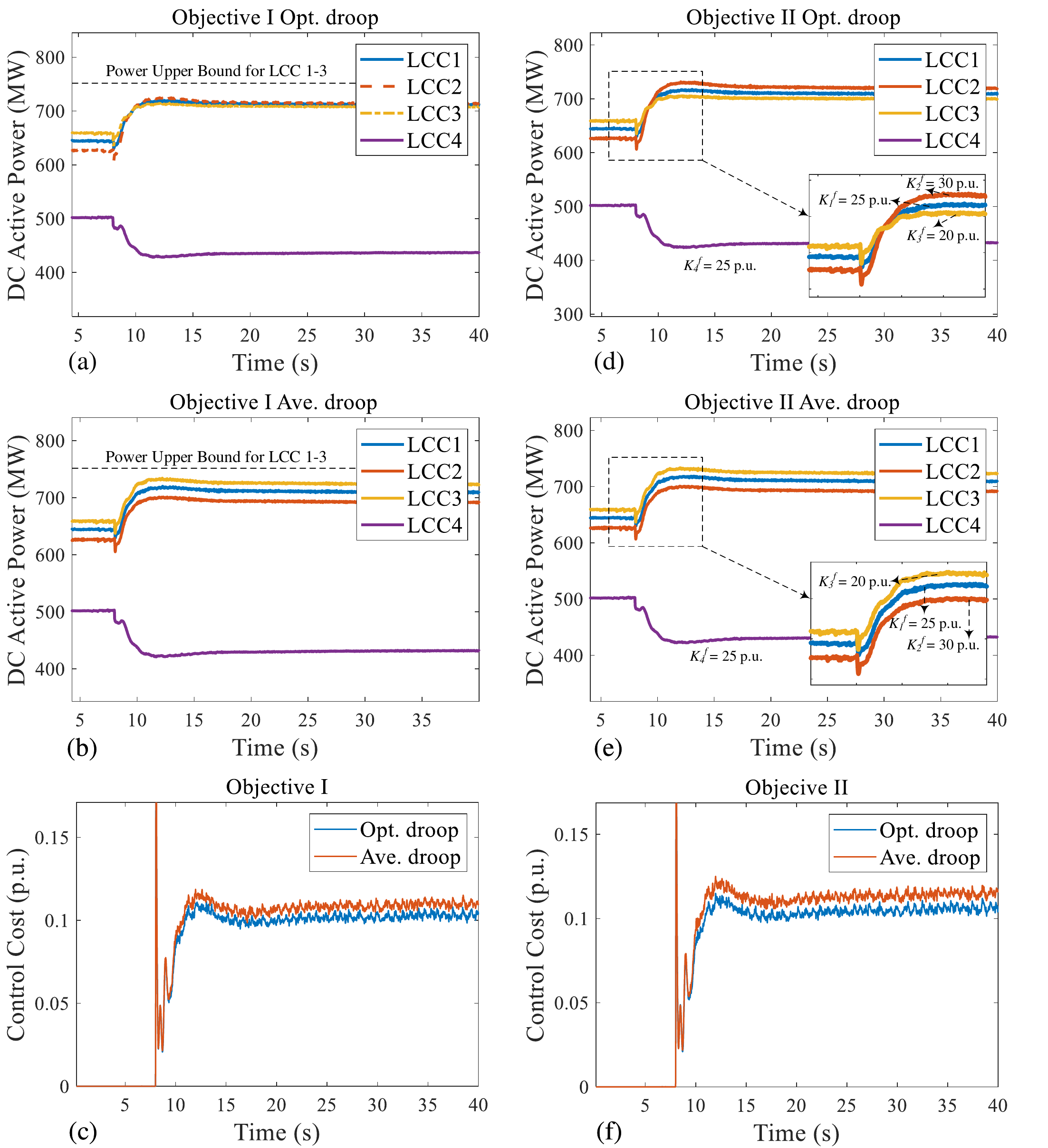}
	\caption{Active powers of LCC-HVDCs and control costs. (a)(b)(c) under control objective I. (d)(e)(f) under control objective II.}
	\label{case_4}
\end{figure}

As for Control Objective I, comparing Fig. \ref{case_4}(a) and \ref{case_4}(b), the steady-state powers of LCC 1-3 almost reach the same level with the optimal droop, while the power variations of the LCC-HVDCs are almost equal with average droop. Thus, due to the same upper bound setting for LCC 1-3, the optimal droop achieves the Control Objective I, i.e., the LCC-HVDC with larger regulation margin provides more power support. By Fig. \ref{case_4}(c), system has less total control cost with optimal droop, which verifies the optimality of optimal coefficients under Control Objective I.

As for Control Objective II, in Fig. \ref{case_4}(e), the LCC-HVDCs almost provide the equal power support to the AC main system However, in Fig. \ref{case_4}(d), the LCC-HVDC connected to larger-primary-frequency-regulation-coefficient adjacent system provides more power support, which meets the Control Objective II. In Fig. \ref{case_4}(f), the control costs remain zero before the fault occurs, and the post-fault control cost with optimal droop coefficients is smaller than that with average droop, which also shows the optimality.

\subsection{Discussion on System Stability}

In Section V, we obtain the sufficient condition for the asymptotic stability of the closed-loop equilibrium, i.e., $k_i^D>0, i\in \mathcal{N}_D$. To further discuss the system stability, five groups of droop coefficients between 1 p.u. to 100 p.u. are selected and set to all LCC-HVDC systems. The AC main system frequencies are shown in Fig. \ref{case_5}.
\begin{figure}[htb]
	\centering
	\includegraphics[width=0.48\textwidth]{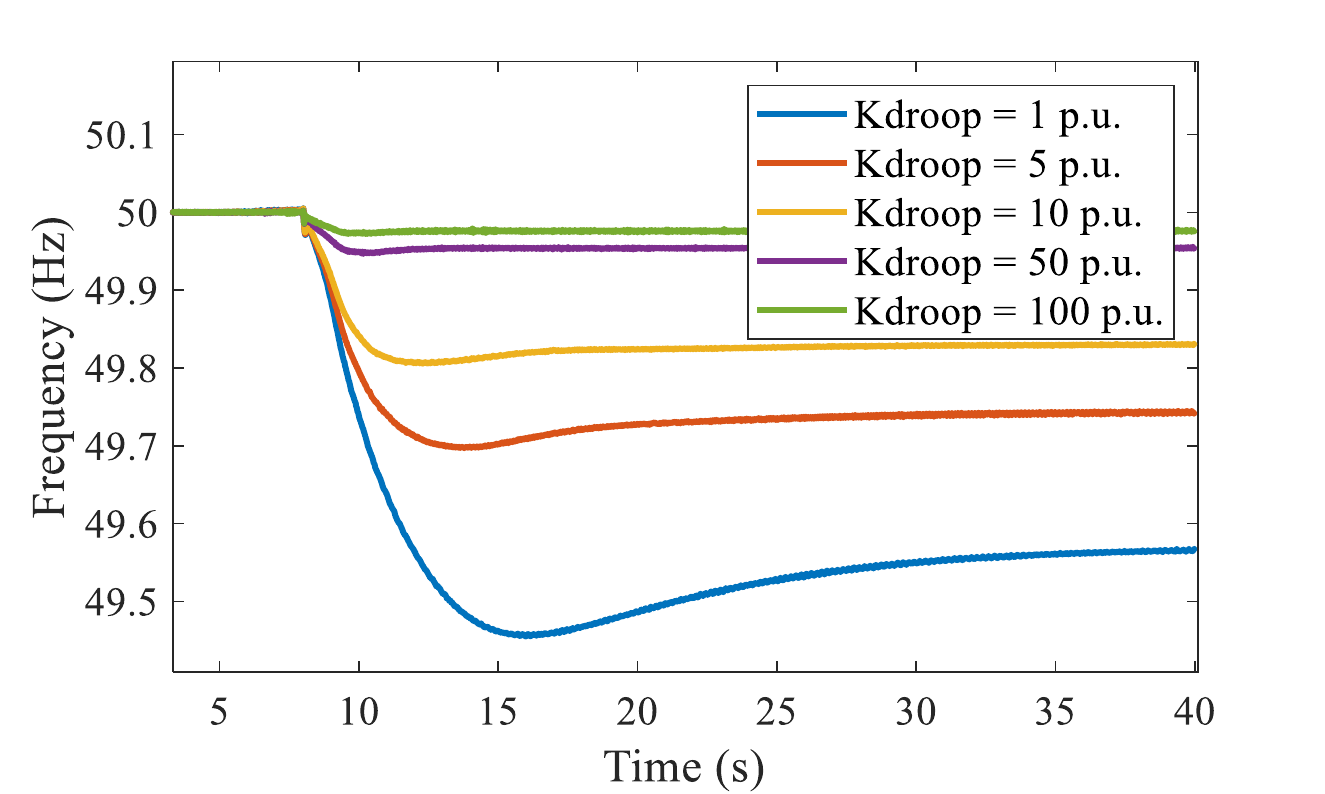}
	\caption{Frequency of AC main system with different droop coefficients}
	\label{case_5}
\end{figure}

As shown in Fig. \ref{case_5}, with different droop coefficients between 1 p.u. to 100 p.u., the system equilibrium can remain stable according to the frequencies. As the droop coefficients increase, the system frequency can be stabilized at a higher level and the transient response time is shorter. Nevertheless, the attractive domain of the equilibrium is not derived in this paper, which is related to the amount of power imbalance, and the system could be frequency instable if there are severer emergency faults occurring. 

\section{Conclusion}

In this paper, a decentralized coordinated-droop-based emergency frequency control strategy is proposed to deal with the emergency frequency instability in MIDC systems. The introduced P-f droop characteristic of LCC-HVDC system enables this control strategy. With the designed coordinated droop mechanism and the dead zone setting, the LCC-HVDC droop controllers work only in case of emergency situations. Benefit from the decentralized control approach, the proposed control strategy is free from controllers' communication and can respond quickly. Then, in order to reasonably allocate the power imbalance among LCC-HVDCs and generators, the optimal droop coefficients are determined by formulating the OEFC problem, which is applicable to various control objectives. Moreover, the locally asymptotic stability of the closed-loop equilibrium is proved by the Lyapunov approach. In the case study on the CloudPSS platform, the effectiveness of the proposed control strategy and the optimality of the selected optimal droop coefficients are verified, and the system stability is illustrated.

%
%




\bibliographystyle{IEEEtran}
\bibliography{mybib}






%
%
%
%
%
%
%

\end{document}